\documentclass{amsart}
\usepackage{amsmath,amssymb,bbm}
\usepackage{amsthm}
\usepackage{graphicx}
\usepackage{subfigure}
\usepackage{enumerate}

\usepackage[colorlinks=true,linkcolor=blue,citecolor=blue,pdfpagelabels=false]{hyperref}

\theoremstyle{plain}
  \newtheorem{theorem}{Theorem}[section]
  \newtheorem{lemma}[theorem]{Lemma}
  
  \newtheorem{corollary}[theorem]{Corollary}
  \newtheorem{proposition}[theorem]{Proposition}
\theoremstyle{definition}
  \newtheorem{example}[theorem]{Example}
  \newtheorem{remark}[theorem]{Remark}
  
  \newtheorem{problem}[theorem]{Problem}
\newenvironment{acknowledgements}{\bigskip\textbf{Acknowledgements.}}{}

\newcommand{\qede}{\hspace*{\fill}$\Diamond$\medskip}
\newcommand{\op}[1]{\ensuremath{\operatorname{#1}}}
\renewcommand{\geq}{\geqslant}

\renewcommand{\leq}{\leqslant}

\newcommand{\pFq}[5]{\ensuremath{{}_{#1}F_{#2} \left( \genfrac{}{}{0pt}{}{#3}{#4} \bigg| {#5} \right)}}
\newcommand{\RR}{\mathbb{R}}

\newcommand{\assign}{:=}
\newcommand{\dueto}[1]{\textup{\textbf{(#1) }}}
\newcommand{\mathd}{\mathrm{d}}

\title{On lattice sums and Wigner limits}

\author{David Borwein}
\address{Department of Mathematics,
Western University,
Middlesex College,
London, Ontario N6A 3K7,
Canada}
\email{dborwein@uwo.ca}

\author{Jonathan M. Borwein}
\address{School of Mathematical and Physical Sciences,
University of Newcastle,
University Drive,
Callaghan, NSW 2308,
Australia}
\email{jonathan.borwein@newcastle.edu.au}

\author{Armin Straub}
\address{Department of Mathematics,
University of Illinois at Urbana-Champaign,
1409 W. Green St,
Urbana, IL 6180,
United States}
\curraddr{Max-Planck-Institut f\"ur Mathematik,
Vivatsgasse 7,
53111 Bonn,
Germany}
\email{astraub@illinois.edu}

\date{August 28, 2013}

\begin{document}

\begin{abstract}
  Wigner limits are given formally as the difference between a lattice sum, associated to a
  positive definite quadratic form, and a corresponding multiple integral.  To
  define these limits, which arose in work of Wigner on the energy of static
  electron lattices, in a mathematically rigorous way one commonly truncates
  the lattice sum and the corresponding integral and takes the limit along
  expanding hypercubes or other regular geometric shapes.
  We generalize the known  mathematically rigorous two and three dimensional results regarding Wigner
  limits, as laid down in {\cite{latticesums-bbs}}, to integer lattices of
  arbitrary dimension. In doing so, we also resolve a problem posed in
  {\cite[Chapter 7]{latticesums}}.

  For the sake of clarity, we begin by considering the simpler case of cubic
  lattice sums first, before treating the case of arbitrary quadratic forms.
  We also consider limits taken along expanding hyperballs with respect to
  general norms, and connect with classical topics such as Gauss's circle
  problem.  An appendix is included to recall certain properties of Epstein
  zeta functions that are either used in the paper or serve to provide
  perspective.
\end{abstract}

\keywords{lattice sums; analytic continuation; Wigner electron sums}

\maketitle

\section{Introduction}\label{sec:intro}

Throughout this paper,
$Q(x)=Q (x_1, \ldots, x_d)$ is a positive definite quadratic form in
$d$ variables with real coefficients
and determinant $\Delta > 0$.
As proposed in {\cite[Chapter 7]{latticesums}}, we shall examine the behaviour of
\[ \sigma_N (s) := \alpha_N (s) - \beta_N (s) \]
as $N \rightarrow \infty$, where $\alpha_N$ and $\beta_N$ are given by
\begin{eqnarray}
  \alpha_N (s) & := & \sum_{n_1 = - N}^N \cdots \sum_{n_d = - N}^N \frac{1}{Q
  (n_1, \ldots, n_d)^s},  \label{eq:alpha}\\
  \beta_N (s) & := & \int_{- N - 1 / 2}^{N + 1 / 2} \cdots \int_{- N - 1 /
  2}^{N + 1 / 2} \frac{\mathd x_1 \cdots \mathd x_d}{Q (x_1, \ldots, x_d)^s} .
  \label{eq:beta}
\end{eqnarray}
As usual, the summation in \eqref{eq:alpha} is understood to avoid the term
corresponding to $(n_1, \ldots, n_d) = (0, \ldots, 0)$. If $\op{Re} s > d /
2$, then $\alpha_N (s)$ converges to the Epstein zeta function $\alpha (s) =
Z_Q (s)$ as $N \rightarrow \infty$. A few basic properties of $Z_Q$ are
recollected in Section \ref{sec:basics}. On the other hand, each integral
$\beta_N (s)$ is only defined for $\op{Re} s < d / 2$.

\textit{A priori}
it is therefore unclear, for any $s$, whether the \emph{Wigner limit} $\sigma (s) \assign
\lim_{N \rightarrow \infty} \sigma_N (s)$ should exist. In the sequel, we will
write $\sigma_Q (s)$ when we wish to emphasize the
dependence on the quadratic form $Q$. For more on the physical background, which motivates the interest in the limit $\sigma
(s)$, we refer to Section \ref{sec:phys} below.

In the case $d = 2$, it was shown in {\cite[Theorem 1]{latticesums-bbs}} that
the limit $\sigma (s)$ exists in the strip $0 < \op{Re} s < 1$ and that it
coincides therein with the analytic continuation of $\alpha (s)$.
Further, in the case $d = 3$ with $Q (x) = x_1^2 + x_2^2 +
x_3^2$, it was shown in {\cite[Theorem 3]{latticesums-bbs}} that the limit
$\sigma (s)$ exists for $1 /
2 < \op{Re} s < 3 / 2$ as well as for $s = 1 / 2$. However, it was
determined that $\sigma (1 / 2) - \pi / 6 = \lim_{\varepsilon \rightarrow 0^+}
\sigma (1 / 2 + \varepsilon)$. In other words, the limit $\sigma (s)$ exhibits
a jump discontinuity at $s = 1 / 2$.

It is therefore natural to ask in what senses the phenomenon, observed for the cubic lattice when $d = 3$, extends
both to higher dimensions and to more general quadratic forms. We largely resolve the following problem which is a refinement of one posed in the recent book {\cite[Chapter 7]{latticesums}}.

\begin{problem}[Convergence] \label{prob}For dimension $d >1$, consider $\sigma_N$ as above.

 \begin{quote}Show that the limit $\sigma(s) \assign \lim_{N \rightarrow
  \infty} \sigma_N (s)$ exists in the strip $d / 2 - 1 < \op{Re} s < d / 2$.
  Does the limit exist for $s = d / 2 - 1$? If so, is the limit discontinuous
  at $s = d / 2 - 1$, and can the height of the jump discontinuity be
  evaluated?\end{quote}
\end{problem}

In Proposition \ref{prop:sigmastrip}, we show that the limit indeed exists in
the strip suggested in Problem \ref{prob}. In the case of $Q (x)
:= x_1^2 + \cdots + x_d^2$, we then show in Theorem \ref{thm:jump} that $\sigma(s)$ also converges for $s = d / 2 - 1$. As in the case $d = 3$, we find that
$\sigma(s)$ has a jump discontinuity, which we evaluate in closed form.  In Theorem \ref{thm:jumpx} we extend this result less explicitly to arbitrary positive definite quadratic forms $Q$.

\subsection{Motivation and physical background}\label{sec:phys}

As described in {\cite[Chapter 7]{latticesums}}:

{\small
\begin{quote}
  In 1934 Wigner introduced the concept of an electron gas bathed in a
  compensating background of positive charge as a model for a metal. He
  suggested that under certain circumstances the electrons would arrange
  themselves in a lattice, and that the body-centred lattice would be the most
  stable of the three common cubic structures. Fuchs (1935) appears to have
  confirmed this in a calculation on copper relying on physical properties of
  copper. The evaluation of the energy of the three cubic electron lattices
  under precise conditions was carried out by Coldwell-Horsefall and Maradudin
  (1960) and became the standard form for calculating the energy of static
  electron lattices. In this model electrons are assumed to be negative point
  charges located on their lattice sites and surrounded by an equal amount of
  positive charge uniformly distributed over a cube centered at the lattice
  point.
\end{quote}}

In three dimensions, this leads precisely to the problem enunciated in the
previous section. That is, Wigner, when $d=3$, $s=1/2$ and $Q(x)=x_1^2+x_2^2+x_3^2$, proposed considering, after appropriate renormalization, the entity
\begin{equation}
\sigma(s):= \sum_{n_1 = - \infty}^\infty \cdots \sum_{n_d = - \infty}^\infty \frac{1}{Q
  (n_1, \ldots, n_d)^s}- \int_{- \infty}^{\infty} \cdots \int_{-\infty}^{\infty} \frac{\mathd x_1 \cdots \mathd x_d}{Q (x_1, \ldots, x_d)^s} .
  \label{eq:wigner}
\end{equation}
As a physicist Wigner found it largely untroubling that in \eqref{eq:wigner} the object of study $\sigma(s)$ never makes unambiguous sense. Nor even does it point the way to formalize its mathematical content.
The concept is thus both natural physically and puzzling mathematically for the reasons given above
of the non-convergence of the integral whenever the sum converges.

 The best-behaved case is that of two dimensions, which is also physically meaningful if used to consider planar lamina.
In  \cite[\S3]{latticesums-bbsz} a `meta-principle' was presented justifying  the evaluation of $\sigma$ as the analytic continuation of $\alpha=Z_Q$. This was followed by a discussion and analysis of various important hexagonal and diamond, cubic and triangular lattices in two and three space  \cite[\S4]{latticesums-bbsz}. In particular, the values  of $\alpha$ obtained agreed with values in the physical literature whenever they were known. It was this work which led to the analysis  in \cite{latticesums-bbs}.

The entirety of \cite[Chapter 7]{latticesums} is
dedicated to the analysis of such `electron sums' in two and three dimensions.
While we make no direct claim for the physical relevance of the analysis with $d>3$, the delicacy of the mathematical resolution of Problem \ref{prob} is certainly informative even just for general forms in three dimensions.

\subsection{Structure of the paper}

The remainder of the paper is organized as follows. In Section \ref{sec:basics},
we establish some basic properties of $\alpha_N$ and $\beta_N$.
Then, in Section \ref{sec:wigconv}, we establish convergence in the strip for a general
quadratic form (Proposition \ref{prop:sigmastrip}). Next, in Section \ref{sec:wigjump},
we consider convergence on the boundary of the strip. In particular, we explicitly evaluate the jump
discontinuity in the cubic case (Theorem \ref{thm:jump}). In the non-cubic
case the same phenomenon is established, though the corresponding evaluation of the jump is less explicit (Theorem \ref{thm:jumpx}).
In Section \ref{sec:alt}, we consider other limiting procedures which
replace limits over expanding cubes by more general convex bodies. The paper
concludes with a brief accounting of the underlying theory of cubic lattice
sums in Appendix \ref{sec:cubiclatticesums}. For more details the reader is
referred to \cite{latticesums} and the other cited works.

\section{Basic analytic properties}\label{sec:basics}

Any quadratic form $Q(x)=Q (x_1, \ldots, x_d)$
can be expressed as
\begin{equation}\label{eq:Q}
  Q (x) = Q_A (x) \assign x^T A x = \sum_{1 \leq i, j \leq d}
  a_{i j} x_i x_j,
\end{equation}
for a matrix $A = (a_{i j})_{1 \leq i, j \leq d}$ which is symmetric
(that is, $a_{i j} = a_{j i}$ for all $1 \leq i, j \leq d$).
If $Q$ is positive definite, then $A$ is a positive definite matrix of
determinant $\Delta = \det(A) > 0$. A basic property of a positive definite
matrix $A$, given in most linear algebra texts, is that it can be decomposed as
$A=L^TL$, where $L$ is a non-singular matrix. This property is used
implicitly when making coordinate transformations as in \eqref{eq:intQQ} and in
the proof of Lemma \ref{lem:intQAB} below.

As indicated in the introduction, the limit of $\alpha_N (s)$ is the Epstein
zeta function
\begin{equation}
 \alpha(s) \assign Z_Q (s) := \sum_{n_1, \ldots, n_d}' \frac{1}{Q (n_1, n_2,
  \ldots, n_d)^s} \label{eq:def:epsteinzeta} .
\end{equation}
Standard arguments show that $Z_Q (s)$ is an analytic function in the domain
$\op{Re} s > d / 2$. In fact, see {\cite{epstein-zeta}} or {\cite[Chapter
2 or 8]{berndtknopp-hecke}}, the Epstein zeta function $Z_Q (s)$ has a meromorphic
continuation to the entire complex plane and satisfies the functional equation
\begin{equation}
  \frac{Z_Q (s) \Gamma (s)}{\pi^s} = \frac{1}{\sqrt{\Delta}}  \frac{Z_{Q^{-
  1}} (d / 2 - s) \Gamma (d / 2 - s)}{\pi^{d / 2 - s}},
  \label{eq:epsteinfunc}
\end{equation}
  where $ Q (x) = x^T A x$ and $Q^{-1} (x) = x^T A^{-1} x.$
Moreover, the only pole of $Z_Q (s)$ occurs at $s = d / 2$, is simple, and has
residue
\begin{equation}
  \op{res}_{d / 2} Z_Q (s) = \frac{1}{\sqrt{\Delta}}  \frac{\pi^{d /
  2}}{\Gamma (d / 2)} . \label{eq:epsteinres}
\end{equation}
A particularly important special case of Epstein zeta functions is that of cubic
lattice sums, which correspond to the choice $Q (x) = x_1^2 + \cdots + x_d^2$.
In Appendix \ref{sec:cubiclatticesums}, we recall some of their basic
properties, which provide further context for the questions discussed herein.

These remarks made, it is natural to begin our investigation of Problem
\ref{prob} by discussing some related properties of the limit of $\beta_N
(s)$.  In the sequel, we use the notation $\|x\|_{\infty} := \max (|x_1 |, \ldots,
|x_d |)$ for vectors $x = (x_1, \ldots, x_d) \in \RR^d$.

\begin{proposition}
  \label{prop:betares}Let $Q$ be a $d$-dimensional positive definite quadratic
  form of determinant $\Delta > 0$. The integral $\beta_N (s)$ extends
  meromorphically to the entire complex plane with a single pole at $s = d /
  2$, which is simple and has residue
  \begin{equation}
    \op{res}_{d / 2} \beta_N (s) = - \frac{1}{\sqrt{\Delta}}  \frac{\pi^{d /
    2}}{\Gamma (d / 2)} = - \op{res}_{d / 2} \alpha (s) .
    \label{eq:betaNres}
  \end{equation}
\end{proposition}

\begin{proof}
  The integral $\beta_N (s)$, as defined in \eqref{eq:beta}, is analytic for
  $\op{Re} s < d / 2$. On the other hand, we easily see that the difference
  \begin{equation}
    \beta_N (s) - \int_{Q (x) \leq N} \frac{1}{Q (x)^s} \mathd x
    \label{eq:betadiff}
  \end{equation}
  is an entire function in $s$. The latter integral can be evaluated in closed
  form. Indeed, for $\op{Re} s < d / 2$,
  \begin{eqnarray}\label{eq:intQQ}
    \int_{Q (x) \leq N} \frac{1}{Q (x)^s} \mathd x & = &
    \frac{1}{\sqrt{\Delta}} \int_{\|x\|_2^2 \leq N} \frac{1}{\|x\|^{2 s}_2}
    \mathd x\nonumber\\
    & = & \frac{1}{\sqrt{\Delta}} \op{vol} (\mathbb{S}^{d - 1})
    \int_0^{\sqrt{N}} r^{d - 1 - 2 s} \mathd r\nonumber\\
    & = & \frac{1}{\sqrt{\Delta}} \frac{N^{d / 2 - s}}{d / 2 - s}
    \frac{\pi^{d / 2}}{\Gamma (d / 2)} .
  \end{eqnarray}
  In light of \eqref{eq:betadiff}, this shows that $\beta_N (s)$ has an
  analytic continuation to the full complex plane with a single pole at $s = d
  / 2$, which is simple and has residue as claimed in \eqref{eq:betaNres}. The
  second equality in \eqref{eq:betaNres} follows from \eqref{eq:epsteinres}.
\end{proof}

We note that the fact that the residue of $\beta_N (s)$ does not depend on $N$
reflects that, for any $N, M > 0$, the differences $\beta_N (s) - \beta_M (s)$
are, as in \eqref{eq:betadiff}, entire functions.

We further record that the proof of Proposition \ref{prop:betares} is related
to the following observation. For any reasonable function $F_s : \RR^d
\rightarrow \RR$ such that $F_s (\lambda x) = | \lambda |^{- 2 s} F_s
(x)$,
\begin{eqnarray}
  \int_{\|x\|_{\infty} \leq 1} F_s (x) \mathd \lambda_d & = & \int_0^1
  \int_{\|x\|_{\infty} = t} F_s (x) \mathd \lambda_{d - 1} \mathd t
  \nonumber\\
  & = & \int_0^1 t^{d - 1 - 2 s} \int_{\|x\|_{\infty} = 1} F_s (x) \mathd
  \lambda_{d - 1} \mathd t \nonumber\\
  & = & \frac{1}{d - 2 s} \int_{\|x\|_{\infty} = 1} F_s (x) \mathd \lambda_{d
  - 1},  \label{eq:intFs}
\end{eqnarray}
where $\lambda_d$ denotes the $d$-dimensional Lebesgue measure and $\lambda_{d
- 1}$ the induced $(d - 1)$-dimensional surface measure (that is, $\mathd
\lambda_{d - 1} = \mathd x_1 \cdots \mathd x_{j - 1} \mathd x_{j + 1} \cdots
\mathd x_d$ on the part of the domain where $x_j$ is constant).

\begin{remark}
  Since the notation used in (\ref{eq:intFs}) is rather terse, let us, for
  instance, spell out the crucial first equality. By separating the variable of
  maximal absolute value and then interchanging summation and integration,
  \begin{eqnarray*}
    \int_{\|x\|_{\infty} \leq 1} F_s (x) \mathd \lambda_d & = & 2 \sum_{j
    = 1}^d \int_0^1 \left( \int_{[- x_j, x_j]^{d - 1}} F_s (x) \mathd x_1
    \cdots \mathd x_{j - 1} \mathd x_{j + 1} \cdots \mathd x_d \right) \mathd
    x_j\\
    & = & \int_0^1 \left( \sum_{j = 1}^d \int_{\substack{\|x\|= |t| \\ x_j = \pm t}}
    F_s (x) \mathd x_1 \cdots \mathd x_{j - 1} \mathd x_{j + 1}
    \cdots \mathd x_d \right) \mathd t\\
    & = & \int_0^1 \int_{\|x\|_{\infty} = t} F_s (x) \mathd \lambda_{d - 1}
    \mathd t.
  \end{eqnarray*}
  We note that the relation between first and final integral also holds with
  $\| \cdot \|_{\infty}$ replaced by $\| \cdot \|_2$ (in which case
  $\lambda_{d - 1}$ would now refer to the surface measure on the Euclidean
  sphere $\{x \in \RR^d : \|x\|_2 = t\}$).
  \qede
\end{remark}

Based on \eqref{eq:intFs}, we obtain the following consequence of Proposition
\ref{prop:betares}, which will be important for our purposes later on.

\begin{lemma}
  \label{lem:intQdi}Let $Q$ be a $d$-dimensional positive definite quadratic
  form of determinant $\Delta > 0$. Then we have
  \begin{equation}
    \int_{\|x\|_{\infty} = 1} \frac{1}{Q (x)^{d / 2}} \mathd \lambda_{d - 1} =
    \frac{2}{\sqrt{\Delta}}  \frac{\pi^{d / 2}}{\Gamma (d / 2)} .
    \label{eq:intQdi}
  \end{equation}
\end{lemma}

\begin{proof}
  From the arguments in the proof of Proposition \ref{prop:betares}, we know
  that
  \[ \int_{\|x\|_{\infty} \leq 1} \frac{1}{Q (x)^s} \mathd \lambda_d \]
  is a meromorphic function with a simple pole at $s = d / 2$. The computation
  in \eqref{eq:intFs} shows that
  \[ \int_{\|x\|_{\infty} = 1} \frac{1}{Q (x)^{d / 2}} \mathd \lambda_{d - 1}
     = - 2 \op{res}_{d / 2} \int_{\|x\|_{\infty} \leq 1} \frac{1}{Q
     (x)^s} \mathd \lambda_d = \frac{2}{\sqrt{\Delta}}  \frac{\pi^{d /
     2}}{\Gamma (d / 2)}, \]
  with the last equality following in analogy with Proposition
  \ref{prop:betares}.
\end{proof}

\begin{example}[Generalized $\arctan(1)$]
  The special case $Q (x) := x_1^2 + \cdots + x_d^2$ results in the integral
  evaluation
  \begin{equation}
    \int_{[- 1, 1]^{d - 1}} \frac{1}{(1 + x_1^2 + \cdots x_{d - 1}^2)^{d / 2}}
    \mathd x = \frac{1}{d}  \frac{\pi^{d / 2}}{\Gamma (d / 2)},
    \label{eq:arctan1}
  \end{equation}
  which has been derived in {\cite[Section 5.7.3]{bb}} as a radially invariant
  generalization of $\arctan (1)$.

  Let us indicate an alternative direct derivation of \eqref{eq:arctan1}. To
  this end, recall that the gamma function is characterized by
  \[ \frac{\Gamma (s)}{A^s} = \int_0^{\infty} t^{s - 1} e^{- A t} \mathd t. \]
  Applying this integral representation, which is valid for $\op{Re} s > 0$,
  with $A = 1 + x_1^2 + \cdots x_{d - 1}^2$, we find
  \begin{eqnarray*}
    &  & \int_{[- 1, 1]^{d - 1}} \frac{1}{(1 + x_1^2 + \cdots x_{d - 1}^2)^s}
    \mathd x\\
    & = & \frac{1}{\Gamma (s)} \int_{[- 1, 1]^{d - 1}} \int_0^{\infty} t^{s -
    1} e^{- (1 + x_1^2 + \cdots x_{d - 1}) t} \mathd t \mathd x\\
    & = & \frac{1}{\Gamma (s)} \int_0^{\infty} t^{s - 1} e^{- t} \int_{[- 1,
    1]^{d - 1}} e^{- (x_1^2 + \cdots x_{d - 1}) t} \mathd x \mathd t\\
    & = & \frac{1}{\Gamma (s)} \int_0^{\infty} t^{s - 1} e^{- t} \left(
    \int_{- 1}^1 e^{- x^2 t} \mathd x \right)^{d - 1} \mathd t\\
    & = & \frac{2}{\Gamma (s)} \int_0^{\infty} u^{2 s - 1} e^{- u^2} \left(
    \int_{- 1}^1 e^{- x^2 u^2} \mathd x \right)^{d - 1} \mathd u.
  \end{eqnarray*}
  In particular, for $s = d / 2$,
  \begin{eqnarray*}
    \int_{[- 1, 1]^{d - 1}} \frac{1}{(1 + x_1^2 + \cdots x_{d - 1}^2)^{d / 2}}
    \mathd x & = & \frac{2}{\Gamma (d / 2)} \int_0^{\infty} e^{- u^2} \left( u
    \int_{- 1}^1 e^{- x^2 u^2} \mathd x \right)^{d - 1} \mathd u.
  \end{eqnarray*}
  Define
  \[ f (u) = u \int_{- 1}^1 e^{- x^2 u^2} \mathd x, \]
  which, in terms of the error function, can be expressed as $f (u) =
  \sqrt{\pi} \op{erf} (u)$. We note that $f (u) \rightarrow \sqrt{\pi}$ as
  $u \rightarrow \infty$. Further, the derivative is simply
  \[ f' (u) = 2 e^{- u^2} . \]
  After the substitution $v = f (u)$, we thus find
  \[ \int_{[- 1, 1]^{d - 1}} \frac{1}{(1 + x_1^2 + \cdots x_{d - 1}^2)^{d /
     2}} \mathd x = \frac{1}{\Gamma (d / 2)} \int_0^{\sqrt{\pi}} v^{d - 1}
     \mathd v = \frac{1}{d}  \frac{\pi^{d / 2}}{\Gamma (d / 2)}, \]
  as claimed.
  \qede
\end{example}

As in \eqref{eq:Q}, we denote with $Q=Q_A$ the quadratic form $Q (x)$ on
$\RR^d$ associated with the symmetric matrix $A = (a_{i j})_{1 \leq i,
j \leq d}$.
We now record an extension of Lemma \ref{lem:intQdi}, which will prove useful
for our purposes (the case $B = A$ in \eqref{eq:intQAB} reduces to
\eqref{eq:intQdi}, and it is the case $B = A^2$ that will appear later).
Recall that the \emph{trace} of a square matrix is given by $\op{tr}A =
\sum_{j=1}^d a_{jj}$ and defines a Euclidean norm on the symmetric $d \times d$
matrices via $\langle A_1,A_2 \rangle = \op{tr} (A_1A_2)$.

\begin{lemma}
  \label{lem:intQAB}For matrices $A, B \in \RR^{d \times d}$, with $A$
  positive definite,
  \begin{equation}
    \int_{\|x\|_{\infty} = 1} \frac{Q_B (x)}{Q_A (x)^{1 + d / 2}} \mathd
    \lambda_{d - 1} = \frac{\op{tr} (B A^{- 1})}{\sqrt{\det (A)}}
    \frac{\pi^{d / 2}}{\Gamma (1 + d / 2)} . \label{eq:intQAB}
  \end{equation}
\end{lemma}

\begin{proof}
  On decomposing $A$ as $A = L^T L$, we find
  \[ \int_{Q (x) \leq 1} \frac{Q_B (x)}{Q_A (x)^{s + 1}} \mathd x =
     \frac{1}{\det (L)} \int_{\|x\|_2^2 \leq 1} \frac{Q_C (x)}{\|x\|^{2
     s + 2}_2} \mathd x, \]
  with $C \assign (L^{- 1})^T B L^{- 1}$. For the residue of the latter integral
  only the quadratic terms $C_{11} x_1^2 + \cdots + C_{d d} x_d^2$ in $Q_C
  (x)$ contribute; indeed, in the present case the contributions of the mixed
  terms $C_{i j} x_i x_j$, $i \neq j$, integrate to zero. Because of symmetry
  we thus obtain
  \begin{eqnarray*}
    \int_{Q (x) \leq 1} \frac{Q_B (x)}{Q_A (x)^{s + 1}} \mathd x & = &
    \frac{\op{tr} (C)}{\sqrt{\det (A)}} \int_{\|x\|_2^2 \leq 1}
    \frac{x_1^2}{\|x\|^{2 s + 2}_2} \mathd x\\
    & = & \frac{\op{tr} (C)}{d \sqrt{\det (A)}} \int_{\|x\|_2^2 \leq
    1} \frac{1}{\|x\|^{2 s}_2} \mathd x\\
    & = & \frac{\op{tr} (C)}{d \sqrt{\det (A)}} \frac{1}{d / 2 - s}
    \frac{\pi^{d / 2}}{\Gamma (d / 2)},
  \end{eqnarray*}
  with the final step as in \eqref{eq:intQQ}. Since the trace is commutative,
  \[ \op{tr} (C) = \op{tr} (L^{- T} B L^{- 1}) = \op{tr} (B L^{-
     1} L^{- T}) = \op{tr} (B A^{- 1}). \]
  We conclude that, for any compact region $D \subset \RR^d$
  containing a neighborhood of the origin,
  \begin{equation}
    \op{res}_{d / 2} \int_D \frac{Q_B (x)}{Q_A (x)^{s + 1}} \mathd x = -
    \frac{\op{tr} (B A^{- 1})}{d \sqrt{\det (A)}}  \frac{\pi^{d / 2}}{\Gamma
    (d / 2)} . \label{eq:resintQAB}
  \end{equation}
  In light of the computation \eqref{eq:intFs}, we arrive at
  \[ \int_{\|x\|_{\infty} = 1} \frac{Q_B (x)}{Q_A (x)^{1 + d / 2}} \mathd
     \lambda_{d - 1} = - 2 \op{res}_{d / 2} \int_{\|x\|_{\infty} \leq
     1} \frac{Q_B (x)}{Q_A (x)^{s + 1}} \mathd x, \]
  which, together with \eqref{eq:resintQAB}, implies \eqref{eq:intQAB}.
\end{proof}

\section{Convergence of Wigner limits}\label{sec:wigconv}

Our next goal is to show that $\sigma_N (s)$ indeed converges in the vertical
strip suggested in Problem \ref{prob}.  As discussed in \cite[Chapter 2 and
8]{latticesums}, convergence over such hyper-cubes is more stable than that
over Euclidean balls and similar shapes. Other limit procedures are compared in Section
\ref{sec:alt}.

\begin{proposition}[Convergence in a strip]
  \label{prop:sigmastrip}Let $Q$ be an arbitrary positive definite quadratic form on $\RR^d$. Then
  the limit $\sigma (s) \assign \lim_{N \rightarrow \infty} \sigma_N (s)$
  exists in the strip $d / 2 - 1 < \op{Re} s < d / 2$ and coincides therein
  with the analytic continuation of $\alpha (s)$.
\end{proposition}

\begin{proof}
  For the first part of the claim, we follow the proof given in
  {\cite{latticesums-bbs}} for binary forms $Q$. Fix $\sigma > 0$ as well as
  $R > 0$ and set $\Omega := \{s : \op{Re} s > \sigma, \hspace{1em} |s| <
  R\}$. All order terms below are uniform with respect to $s$ in the bounded
  region $\Omega$. For $N \geq 1$ let
  \begin{eqnarray*}
    \delta_N (s) & := & \sigma_N (s) - \sigma_{N - 1} (s)\\
    & = & \sum_{\|n\|_{\infty} = N} \int_{\|x\|_{\infty} \leq 1 / 2}
    \left[ \frac{1}{Q (n)^s} - \frac{1}{Q (n + x)^s} \right] \mathd x.
  \end{eqnarray*}

  Here and in the sequel, we let $f (x) := Q (n + x)^{- s}$ with $\|n\|_{\infty} = N$ and $\|x\|_{\infty}
  \leq 1 / 2$. Since we may assume $Q (x) = \sum_{i, j} a_{i j} x_i x_j$, with $a_{i j}
  = a_{j i}$, is positive definite, we have the estimate
  \begin{eqnarray*}
    f_{i j} (x) & = & \frac{4 s (s + 1)}{Q (n + x)^{s + 2}} \sum_k a_{i k}
    (n_k + x_k) \sum_{\ell} a_{j \ell} (n_{\ell} + x_{\ell}) - \frac{2 a_{i j}
    s}{Q (n + x)^{s + 1}}\\
    & = & O (N^{- 2 \sigma - 2}).
  \end{eqnarray*}
  Here, the indices of $f$ indicate partial derivatives with respect to the $i$-th or $j$-th variable.
  We thus have
  \begin{equation}
    f (x) - f (0) = \sum_i x_i f_{i} (0) + O (N^{- 2 \sigma - 2}).
    \label{eq:fx0}
  \end{equation}
  Consequently,
  \begin{eqnarray*}
    &  & \int_{\|x\|_{\infty} \leq 1 / 2} \left[ \frac{1}{Q (n)^s} -
    \frac{1}{Q (n + x)^s} \right] \mathd x\\
    & = & - \int_{\|x\|_{\infty} \leq 1 / 2} \left[ \sum_i x_i f_{i}
    (0) \right] \mathd x + O (N^{- 2 \sigma - 2})\\
    & = & O (N^{- 2 \sigma - 2}),
  \end{eqnarray*}
  because the final integral, being odd, vanishes.

  Hence,
  \[ \delta_N (s) = O (N^{d - 2 \sigma - 3}), \]
  and so, for all $s \in \Omega$, $| \delta_N (s) | \leq M N^{d - 2
  \sigma - 3}$ for some $M$, which is independent of $N$ and $s$. Assume now
  that $\sigma > d / 2 - 1$. Since $\delta_N (s)$ is an entire function, the
  Weierstrass $M$-test shows that
  \[ \delta (s) := \sum_{N = 1}^{\infty} \delta_N (s) \]
  is an analytic function in $\Omega$. Since $R$ was arbitrary, $\delta (s)$
  is in fact analytic in the half-plane $\op{Re} s > d / 2 - 1$. By
  construction,
  \begin{equation}
    \delta (s) = \lim_{N \rightarrow \infty} \left[ \sigma_N (s) - \sigma_0
    (s) \right] = \lim_{N \rightarrow \infty} \left[ \sigma_N (s) + \beta_0
    (s) \right] . \label{eq:delta1}
  \end{equation}
  It follows that the limit $\sigma (s)$ exists if, additionally, $\op{Re} s
  < d / 2$.

  \medskip

  For the second part of the claim, we begin with the simple observation that,
  for $\op{Re} s < d / 2$,
  \begin{eqnarray}
    \beta_N (s) & = & \int_{\|x\|_{\infty} \leq N + 1 / 2} \frac{1}{Q
    (x)^s} \mathd x \nonumber\\
    & = & (2 N + 1)^{d - 2 s} \int_{\|x\|_{\infty} \leq 1 / 2}
    \frac{1}{Q (x)^s} \mathd x \nonumber\\
    & = & (2 N + 1)^{d - 2 s} \beta_0 (s) .  \label{eq:betaN0}
  \end{eqnarray}
  As shown in Proposition \ref{prop:betares}, both $\beta_N$ and $\beta_0$
  have meromorphic extensions to the entire complex plane, and the relation
  \eqref{eq:betaN0} continues to hold. In particular, this shows that, for
  $\op{Re} s > d / 2$, the meromorphic continuation of $\beta_N$ satisfies
  \[ \lim_{N \rightarrow \infty} \beta_N (s) = \lim_{N \rightarrow \infty} (2
     N + 1)^{d - 2 s} \beta_0 (s) = 0. \]
  Working from \eqref{eq:delta1}, we thus have, for $\op{Re} s > d / 2$,
  \begin{equation}
    \delta (s) = \lim_{N \rightarrow \infty} \left[ \alpha_N (s) - \beta_N (s)
    + \beta_0 (s) \right] = \alpha (s) + \beta_0 (s) . \label{eq:deltaab}
  \end{equation}
  On the other hand, we have shown via \eqref{eq:delta1} that, for $\op{Re}
  s < d / 2$,
  \begin{equation}
    \delta (s) = \sigma (s) + \beta_0 (s) . \label{eq:deltasb}
  \end{equation}
  Since both $\delta (s)$ and $\beta_0 (s)$ are meromorphic in the half-plane
  $\op{Re} s > d / 2 - 1$, comparing \eqref{eq:deltaab} and
  \eqref{eq:deltasb} proves that the analytic continuations of $\sigma (s)$
  and $\alpha (s)$ agree. In particular, in the strip $d / 2 - 1 < \op{Re} s
  < d / 2$, the limit $\sigma (s)$, which was shown to exist, equals the
  analytic continuation of $\alpha (s)$.
\end{proof}

We note that Proposition \ref{prop:sigmastrip} agrees with the results known
for $d = 2, 3$. In the $d = 2$ case, the limit $\sigma (s)$ exists for $0 <
\op{Re} s < 1$, in accordance with {\cite[Theorem 1]{latticesums-bbs}}. In
the $d = 3$ case, the limit $\sigma (s)$ exists for $1 / 2 < \op{Re} s < 3 /
2$, which is consistent with the special case of the cubic lattice discussed
in {\cite[Theorem 3]{latticesums-bbs}}.

\section{Jump discontinuities in Wigner limits}\label{sec:wigjump}

In {\cite[Theorem 3]{latticesums-bbs}} it was shown that, in the case of the
cubic lattice, the limit $\sigma (s)$ also exists for $s = 1 / 2$, but is
discontinuous there. In fact, it was shown that \[\sigma (1 / 2) - \pi / 6
= \lim_{\varepsilon \rightarrow 0^+} \sigma (1 / 2 + \varepsilon).\] We now
extend this result to cubic lattices in arbitrary dimensions, in which case
we can and do evaluate the jump discontinuity in simple terms.
We then show that an analogous result is true for arbitrary positive definite quadratic forms,
though the proof is more technical and no simple closed-form expression for the jumps is
given.

\begin{remark}[$\sigma(0)$]
  \label{rk:sigma0}Note that, for trivial reasons, the limit $\sigma (0)$
  always exists and is given by $\sigma (0) = - 1$, which agrees with the
  value $\alpha (0) = - 1$, obtained by analytic continuation from
  \eqref{eq:epsteinfunc} and \eqref{eq:epsteinres}. (The value $s = 0$ is
  missed in the statement of Theorem 3 in {\cite{latticesums-bbs}}.) \qede
\end{remark}

\begin{theorem}[Cubic jump discontinuity]
  \label{thm:jump}Let $Q (x) = x_1^2 + \cdots + x_d^2$. Then the corresponding limit $\sigma
  (s) \assign \lim_{N \rightarrow \infty} \sigma_N (s)$ exists in the strip $d
  / 2 - 1 < \op{Re} s < d / 2$ and for $s = d / 2 - 1$. In the strip,
  $\sigma (s)$ coincides with the analytic continuation of $\alpha (s)$. On
  the other hand,
  \[ \sigma (d / 2 - 1) - \frac{1}{6}  \frac{\pi^{d / 2}}{\Gamma (d / 2 - 1)}
     = \alpha (d / 2 - 1) = \lim_{\varepsilon \rightarrow 0^+} \sigma (d / 2 -
     1 + \varepsilon) . \]
  In particular, for $d \geq 3$, $\sigma (s)$ is discontinuous at $s = d
  / 2 - 1$.
\end{theorem}

\begin{proof}
  In light of Proposition \ref{prop:sigmastrip}, we only need to show the
  statement about the value of $\sigma (s)$ at $s = d / 2 - 1$.

  Let us adopt the notation used in Proposition \ref{prop:sigmastrip},
  including, in particular, the definitions of $\delta_N$ and  $f (x) := Q (n + x)^{- s}$ with $\|n\|_{\infty} = N$ and $\|x\|_{\infty}
  \leq 1 / 2$. Proceeding
  as for \eqref{eq:fx0}, we have that
  \[ f (x) - f (0) = \sum_i x_i f_{i} (0) + \frac{1}{2} \sum_{i, j} x_i x_j
     f_{i j} (0) + O (N^{- 2 \sigma - 3}). \]
  Since terms of odd order in the $x_i$ are eliminated in the subsequent
  integration, we focus on the terms $f_{i i}$. In the present case of the
  cubic lattice,
  \begin{eqnarray}\label{eq:sumfij0}
    \sum_i f_{i i} (0) & = & \sum_i \left[ \frac{4 s (s + 1)}{Q (n)^{s +
    2}} n_i^2 - \frac{2 s}{Q (n)^{s + 1}} \right]\nonumber\\
    & = & \frac{4 s (s + 1) - 2 d s}{Q (n)^{s + 1}}\nonumber\\
    & = & \frac{2 s (2 s - (d - 2))}{Q (n)^{s + 1}} .
  \end{eqnarray}
  We thus find that
  \begin{eqnarray*}
    &  & \int_{\|x\|_{\infty} \leq 1 / 2} \left[ \frac{1}{Q (n)^s} -
    \frac{1}{Q (n + x)^s} \right] \mathd x\\
    & = & - \frac{1}{2} \int_{\|x\|_{\infty} \leq 1 / 2} \left[ \sum_i
    x_i^2 f_{i i} (0) \right] \mathd x + O (N^{- 2 \sigma - 3})\\
    & = & - \frac{1}{24} \sum_i f_{ii} (0) + O (N^{- 2 \sigma - 3}),\end{eqnarray*}
 \noindent (on integrating term-by-term). Then, on appealing to \eqref{eq:sumfij0},
    \begin{eqnarray*} \int_{\|x\|_{\infty} \leq 1 / 2} \left[ \frac{1}{Q (n)^s} -
    \frac{1}{Q (n + x)^s} \right] \mathd x& = & \frac{1}{12}  \frac{s (d - 2 - 2 s)}{Q (n)^{s + 1}} + O (N^{- 2
    \sigma - 3}) .
  \end{eqnarray*}
  Hence,
  \begin{eqnarray*}
    \delta_N (s) & = & \sum_{\|n\|_{\infty} = N} \int_{\|x\|_{\infty}
    \leq 1 / 2} \left[ \frac{1}{Q (n)^s} - \frac{1}{Q (n + x)^s} \right]
    \mathd x\\
    & = & \frac{s (d - 2 - 2 s)}{12} \sum_{\|n\|_{\infty} = N} \frac{1}{Q
    (n)^{s + 1}} + O (N^{d - 2 \sigma - 4})\\
    & = & \frac{s (d - 2 - 2 s)}{12 N^{2 s - d + 3}}  \frac{1}{N^{d - 1}}
    \sum_{\|n\|_{\infty} = N} \frac{1}{Q (n / N)^{s + 1}} + O (N^{d - 2 \sigma
    - 4}) .
  \end{eqnarray*}
  We now note that
  \begin{equation*}
    \frac{1}{2 d N^{d - 1}} \sum_{\|n\|_{\infty} = N} \frac{1}{Q (n /
    N)^{s + 1}}
    = V_N(s) + O (N^{- 1}),
  \end{equation*}
  where
  \begin{equation*}
    V_N(s) :=
    \frac{1}{N^{d - 1}} \sum_{- N \leq n_i < N} \frac{1}{(1 + (n_1
    / N)^2 + \cdots + (n_{d - 1} / N)^2)^{s + 1}}.
  \end{equation*}

  We first show that $V_N
  (s)$ approaches the integral
  \[ V (s) := \int_{[- 1, 1]^{d - 1}} \frac{1}{(1 + x_1^2 + \cdots + x_{d -
     1}^2)^{s + 1}} \mathd x. \]
  Indeed, this follows since, for $\op{Re} s \geq - 2$,
  \begin{eqnarray*}
    |V (s) - V_N (s) | & = & \left| \sum_{- N \leq n_i < N} \int_{n_i
    \leq N x_i \leq n_i + 1} \left[ \frac{1}{(1 + x_1^2 + \cdots +
    x_{d - 1}^2)^{s + 1}} \right. \right.\\
    &  & \left. \left. - \frac{1}{(1 + (n_1 / N)^2 + \cdots + (n_{d - 1} /
    N)^2)^{s + 1}}  \right]\,\mathd x \right|\\
    & \leq & \sum_{- N \leq n_i < N} \int_{n_i \leq N x_i
    \leq n_i + 1} (d - 1) \frac{2 |s + 1|}{N} \,\mathd x\\
    & = & \frac{2^d (d - 1) |s + 1|}{N} .
  \end{eqnarray*}
  To bound the above integrand, we used that $|x^{\lambda} - y^{\lambda} | \leq
  | \lambda | |x - y|$ when $\op{Re} \lambda \leq 1$ and $x, y
  \geq 1$ (as follows from the mean value theorem).

  Combining these estimates, we
  can thus write
  \begin{equation}
    \delta_N (s) = \frac{d}{6}  \frac{s (d - 2 - 2 s)}{N^{2 s - d + 3}} V (s)
    + W_N (s), \label{eq:deltanVW}
  \end{equation}
  where $W_N (s) = O (N^{d - 2 \sigma - 4})$. For $\sigma > d / 2 - 3 / 2$,
  the sum
  \[ W (s) \assign \sum_{N = 1}^{\infty} W_N (s) \]
  converges and, by the Weierstrass $M$-test, defines an analytic function.
  If, further, $\op{Re} s > d / 2 - 1$ then, from \eqref{eq:deltanVW}, the
  sum $\delta (s) \assign \sum_{N = 1}^{\infty} \delta_N (s)$ converges and we
  have
  \[ \delta (s) = \frac{d}{6} s (d - 2 - 2 s) \zeta (2 s - d + 3) V (s) + W
     (s) . \]
  In particular, since $\zeta (s)$ has a simple pole at $s = 1$ of residue
  $1$, we find
\begin{eqnarray}\label{eq:lim}\lim_{\varepsilon \rightarrow 0^+} \delta (d / 2 - 1 + \varepsilon) = -
     \frac{d}{6} (d / 2 - 1) V (d / 2 - 1) + W (d / 2 - 1) . \end{eqnarray}
  On the other hand, it follows from \eqref{eq:deltanVW} that $\delta_N (d / 2
  - 1) = W_N (d / 2 - 1)$.

  Hence, the defining series for $\delta (s)$ also
  converges when $s = d / 2 - 1$ and we obtain
\begin{eqnarray}\label{eq:sigv} \delta (d / 2 - 1) = W (d / 2 - 1) . \end{eqnarray}
  Using the consequence \eqref{eq:arctan1} of Lemma \ref{lem:intQdi}, we have
  \[ \frac{d}{6} (d / 2 - 1) V (d / 2 - 1) = \frac{1}{6}  \frac{\pi^{d /
     2}}{\Gamma (d / 2 - 1)} . \]
  Since, by construction, $\delta (s) = \sigma (s) - \sigma_0 (s) = \sigma
  (s)$, on comparing \eqref{eq:lim} and \eqref{eq:sigv} we are done.
\end{proof}

\begin{example}[Explicit  evaluations in even dimensions]
  In the case of cubic lattice sums and small even dimension, the value
  $\sigma (d / 2 - 1)$, at the jump discontinuity, can be given explicitly by combining Theorem
  \ref{thm:jump} and the closed forms for the corresponding Epstein zeta
  function, recalled in Example \ref{eg:cubicexact} below.
  Let $Q_d(x) = x_1^2+\ldots+x_d^2$.
  \begin{eqnarray*}
    \sigma_{Q_2} (0) & = & \alpha_{Q_2} (0) = - 1,\\
    \sigma_{Q_4} (1) & = & \frac{\pi^2}{6} + \alpha_{Q_4} (1) = \frac{\pi^2}{6} - 8
    \log 2,\\
    \sigma_{Q_6} (2) & = & \frac{\pi^3}{6} + \alpha_{Q_6} (2) = \frac{\pi^3}{6} -
    \frac{\pi^2}{3} - 8 G,\\
    \sigma_{Q_8} (3) & = & \frac{\pi^4}{12} + \alpha_{Q_8} (3) = \frac{\pi^4}{12} - 8
    \zeta (3),\\
    \sigma_{Q_{24}} (11) & = & \frac{\pi^{12}}{6 \cdot 10!} + \alpha_{Q_{24}} (11) =
    \frac{\pi^{12}}{6 \cdot 10!} - \frac{8}{691} \zeta (11) +
    \frac{271435}{5528} L_\Delta (11) .
  \end{eqnarray*}
  Here, $G = \sum_{n = 1}^{\infty} \chi_{- 4} (n) / n^2$ denotes Catalan's
  constant, and $L_\Delta$ is (the analytic continuation of) $L_{\Delta}(s) =
  \sum_{n=1}^\infty \tau(n) / n^s$ with $\tau(n)$ Ramanujan's $\tau$-function.
  A few properties of this remarkable function are commented on in Example
  \ref{eg:cubicexact}. We note that we have used the appropriate reflection formulas to simplify these evaluations.

  We note that the above values mix numbers of different `order', such as
  $\pi^4$ and $\zeta (3)$ which have order $4$ and $3$, respectively.  This may
  be another argument to use $\alpha(d/2-1)$ as the `value' of the Wigner limit even
  when the limit $\sigma(d/2-1)$ itself converges.
  \qede
\end{example}

We now extend Theorem \ref{thm:jump} to arbitrary definite quadratic forms.
For the most part, the proof is a natural extension of the proof of Theorem \ref{thm:jump}.
For the convenience of the reader, we duplicate some parts, as well as the
overall structure, of the previous proof.

As in \eqref{eq:Q}, let $Q = Q_A$ be the positive definite quadratic form
associated to the symmetric matrix $A$. Set
also $B(s) := \op{tr} (A) A - 2 (s + 1) A^2$.  Finally, define
\begin{align}\label{eq:vq}
  V (s) := V_Q(s) := \int_{\|x\|_{\infty} = 1} \frac{Q_{B(s)}
    (x)}{Q_A (x)^{s + 2}} \mathd \lambda_{d - 1},
\end{align}
with $\lambda_{d - 1}$ the induced $(d-1)$-dimensional measure as in \eqref{eq:intFs}.

\begin{theorem}[General jump discontinuity]
  \label{thm:jumpx}Let $Q$ be an arbitrary positive definite quadratic form.
  Then the corresponding limit $\sigma (s) \assign \lim_{N \rightarrow \infty} \sigma_N (s)$
  exists in the strip $d / 2 - 1 < \op{Re} s < d / 2$ and for $s = d / 2 -
  1$. In the strip, $\sigma (s)$ coincides with the analytic continuation of
  $\alpha (s)$. On the other hand,
  \begin{equation}\label{eq:jumpx}
    \sigma (d / 2 - 1) + \frac{d / 2 - 1}{24} V'_Q (d / 2 - 1) = \alpha (d / 2
    - 1) = \lim_{\varepsilon \rightarrow 0^+} \sigma (d / 2 - 1 +
    \varepsilon),
  \end{equation}
  with $V_Q$ as introduced in equation \eqref{eq:vq}.
\end{theorem}

\begin{proof}
  In light of Proposition \ref{prop:sigmastrip}, we
  only need to show the statement about the value of $\sigma (s)$ at $s = d /
  2 - 1$.

  Let us adopt the notation used in Proposition \ref{prop:sigmastrip},
  including, in particular, the definitions of $\delta_N$ and $f (x) \assign Q
  (n + x)^{- s}$ with $\|n\|_{\infty} = N$ and $\|x\|_{\infty} \leq 1 /
  2$. Proceeding as for \eqref{eq:fx0}, we have that
  \[ f (x) - f (0) = \sum_i x_i f_i (0) + \frac{1}{2} \sum_{i, j} x_i x_j f_{i
     j} (0) + O (N^{- 2 \sigma - 3}). \]
  Since terms of odd order in the $x_i$ are eliminated in the subsequent
  integration, we focus on the terms $f_{i i} (0)$, which are given by
  \[ f_{i i} (0) = \frac{4 s (s + 1)}{Q (n)^{s + 2}} \left[ \sum_{k = 1}^d
     a_{i k} n_k \right]^2 - \frac{2 a_{i i} s}{Q (n)^{s + 1}} . \]
  Hence, equation \eqref{eq:sumfij0} generalizes to
  \begin{eqnarray}
    \sum_{i = 1}^d f_{i i} (0) & = & \frac{4 s (s + 1)}{Q (n)^{s + 2}} \sum_{1
    \leq k, l \leq d} \left( \sum_{i = 1}^d a_{k i} a_{i l} \right)
    n_k n_l - 2 s \frac{\op{tr} (A)}{Q (n)^{s + 1}} \nonumber\\
    & = & \frac{4 s (s + 1)}{Q (n)^{s + 2}} Q_{A^2} (n) - 2 s \frac{\op{tr}
    (A)}{Q (n)^{s + 1}} \nonumber\\
    & = & \frac{2 s}{Q (n)^{s + 2}}  \left[ 2 (s + 1) Q_{A^2} (n) - \op{tr}
    (A) Q (n) \right] .  \label{eq:sumfij0x}
  \end{eqnarray}
  We thus find, on integrating term-by-term, that
  \begin{eqnarray*}
    &  & \int_{\|x\|_{\infty} \leq 1 / 2} \left[ \frac{1}{Q (n)^s} -
    \frac{1}{Q (n + x)^s} \right] \mathd x\\
    & = & - \frac{1}{2} \int_{\|x\|_{\infty} \leq 1 / 2} \left[ \sum_{i
    = 1}^d x_i^2 f_{i i} (0) \right] \mathd x + O (N^{- 2 \sigma - 3})\\
    & = & - \frac{1}{24} \sum_{i = 1}^d f_{i i} (0) + O (N^{- 2 \sigma -
    3})\\
    & = & \frac{s}{12}  \frac{\op{tr} (A) Q (n) - 2 (s + 1) Q_{A^2} (n)}{Q
    (n)^{s + 2}} + O (N^{- 2 \sigma - 3}) .
  \end{eqnarray*}
  In the final step, we appealed to \eqref{eq:sumfij0x}. Hence,
  \begin{eqnarray*}
    \delta_N (s) & = & \sum_{\|n\|_{\infty} = N} \int_{\|x\|_{\infty}
    \leq 1 / 2} \left[ \frac{1}{Q (n)^s} - \frac{1}{Q (n + x)^s} \right]
    \mathd x\\
    & = & \frac{s}{12}  \sum_{\|n\|_{\infty} = N} \frac{\op{tr} (A) Q (n) -
    2 (s + 1) Q_{A^2} (n)}{Q (n)^{s + 2}} + O (N^{d - 2 \sigma - 4})\\
    & = & \frac{s}{12 N^{2 s + 2}}  \sum_{\|n\|_{\infty} = N} \frac{\op{tr}
    (A) Q (n / N) - 2 (s + 1) Q_{A^2} (n / N)}{Q (n / N)^{s + 2}} + O (N^{d -
    2 \sigma - 4}) .
  \end{eqnarray*}
  Consider, as defined above, $B(s) = \op{tr} (A) A - 2 (s + 1) A^2$. As in the proof of Theorem
  \ref{thm:jump}, one obtains that, for $\op{Re} s \geq - 2$,
  \begin{align*}
   V (s) = \int_{\|x\|_{\infty} = 1} \frac{Q_{B(s)} (x)}{Q_A (x)^{s + 2}}
     \mathd \lambda_{d - 1} = \frac{1}{N^{d - 1}} \sum_{\|n\|_{\infty} = N}
     \frac{Q_{B(s)} (n / N)}{Q_A (n / N)^{s + 2}} + O (N^{- 1}),
  \end{align*}
  with $\lambda_{d - 1}$ as in \eqref{eq:intFs}. Combining these, we can thus
  write
  \begin{equation}
    \delta_N (s) = s \frac{V (s)}{12 N^{2 s - d + 3}} + W_N (s),
    \label{eq:deltanVWx}
  \end{equation}
  where $W_N (s) = O (N^{d - 2 \sigma - 4})$. For $\sigma > d / 2 - 3 / 2$,
  the sum
  \[ W (s) \assign \sum_{N = 1}^{\infty} W_N (s) \]
  converges and, by the Weierstrass $M$-test, defines an analytic function.
  If, further, $\op{Re} s > d / 2 - 1$ then, from \eqref{eq:deltanVWx}, the
  sum $\delta (s) \assign \sum_{N = 1}^{\infty} \delta_N (s)$ converges and we
  have
  \begin{equation}
    \delta (s) = \frac{s V (s)}{12} \zeta (2 s - d + 3) + W (s) .
    \label{eq:deltaVWx}
  \end{equation}
  Since
  \[ \op{tr} (B(s) A^{- 1}) = \op{tr} (\op{tr} (A) I - 2 (s + 1) A) = (d -
     2 (s + 1)) \op{tr} (A), \]
  Lemma \ref{lem:intQAB} shows that
  \begin{equation}
    V (d / 2 - 1) = 0. \label{eq:V0}
  \end{equation}
  Using that $\zeta (s)$ has a simple pole at $s = 1$ of residue $1$, we thus
  deduce from \eqref{eq:deltaVWx} that
  \begin{equation}\label{eq:deltapx}
    \lim_{\varepsilon \rightarrow 0^+} \delta (d / 2 - 1 + \varepsilon) =
    \frac{d / 2 - 1}{24} V' (d / 2 - 1) + W (d / 2 - 1) .
  \end{equation}
  On the other hand, \eqref{eq:deltanVWx} together with \eqref{eq:V0} implies
  $\delta_N (d / 2 - 1) = W_N (d / 2 - 1)$. Hence, the defining series for
  $\delta (s)$ also converges when $s = d / 2 - 1$ and we obtain
  \[ \delta (d / 2 - 1) = W (d / 2 - 1) . \]
  The claim follows on comparison with \eqref{eq:deltapx}.
\end{proof}

\subsection{The behaviour of $V_Q'(d/2-1)$}

We now examine the nature of $V_Q'(d/2-1)$ in somewhat more detail.
Specifically, we are interested in the following question:
\begin{problem}\label{prob:VQD}
  Let $d>1$. Are there positive definite quadratic forms $Q$ on $\RR^d$ such that
  $V_Q'(d/2-1) = 0$?
\end{problem}
Recall that, in light of Theorem \ref{thm:jumpx}, if $V_Q'(d/2-1) \ne 0$ for a
quadratic form $Q$ on $\RR^d$, with $d>2$, then the corresponding Wigner limit
$\sigma_Q(s)$ exhibits a jump discontinuity at $s=d/2-1$.
In fact, in all the cases of $Q$, that we consider in this section, including
the cubic lattice case, we find that $V_{Q}'(d/2-1) < 0$, which leads us to
speculate whether this inequality holds in general.

From the definition \eqref{eq:vq} we obtain that
\begin{align}\label{eq:VD0}
  V_Q'(d/2-1) &=
  \int_{\|x\|_{\infty} = 1} \frac{-2 Q_{A^2}(x)}
    {Q_A (x)^{d/2 + 1}} \mathd \lambda_{d - 1} \nonumber\\
  &- \int_{\|x\|_{\infty} = 1} \frac{\op{tr}(A) Q_A(x) - d Q_{A^2}(x)}
    {Q_A (x)^{d/2 + 1}} \log Q_A(x) \mathd \lambda_{d - 1} \\
  &= -\frac{4\op{tr} (A)}{d\sqrt{\det (A)}}
    \frac{\pi^{d / 2}}{\Gamma (d / 2)} \nonumber\\
  &- \int_{\|x\|_{\infty} = 1} \frac{\op{tr}(A) Q_A(x) - d Q_{A^2}(x)}
    {Q_A (x)^{d/2 + 1}} \log Q_A(x) \mathd \lambda_{d - 1}. \label{eq:VD}
\end{align}
The last equality is a useful consequence of Lemma \ref{lem:intQAB}.

We also have that
\begin{align}\label{eq:scale}
  V'_{\lambda Q} (d / 2 - 1) = \lambda^{- (d / 2 - 1)} V_Q' (d / 2 -1) .
\end{align}
Indeed, it follows directly from the definition (\ref{eq:vq}) that, for
$\lambda > 0$, $V_{\lambda Q} (s) = \lambda^{- s} V_Q (s)$ and hence, by
(\ref{eq:V0}), that \eqref{eq:scale} holds.
The rescaling result \eqref{eq:scale} shows that scaling $Q$ does not change
the sign of $V_Q' (d / 2 - 1)$.
Also note that both integrals in (\ref{eq:VD0}) scale in the same way; that
this is true for the integral involving the logarithm is equivalent to
\[ \int_{\|x\|_{\infty} = 1} \frac{\op{tr} (A) Q_A (x) - d Q_{A^2}
   (x)}{Q_A (x)^{d / 2 + 1}} \mathd \lambda_{d - 1} = 0, \]
which follows from Lemma \ref{lem:intQAB}.

\begin{example}[Recovery of cubic jump]
  Let us demonstrate that Theorem \ref{thm:jumpx} reduces to Theorem
  \ref{thm:jump} in the cubic lattice case. In that case, $A = I$ and $\op{tr}(A) =d$,
  so that the integral in \eqref{eq:VD}, involving the logarithm, vanishes.
  Hence,
  \[ V' (d / 2 - 1) = - 4 \frac{\pi^{d / 2}}{\Gamma (d / 2)}, \]
  in agreement with the value given in Theorem \ref{thm:jump}.
  \qede
\end{example}

We now give a simple criterion that $V_Q' (d / 2 - 1) < 0$ for certain $Q =
Q_A$.  Suppose that there is some $\lambda>0$ such that, for all $x$ with
$\|x\|_{\infty} = 1$,
\begin{equation}
  2 Q_{A^2} (x) \geq \left[ d Q_{A^2} (x) - \op{tr} (A) Q_A (x) \right]
\, \log Q_{\lambda A} (x). \label{eq:Qineq}
\end{equation}
It then follows from \eqref{eq:VD0} that $V_Q' (d / 2 - 1) \leq 0$. To see that,
in fact, $V_Q' (d / 2 - 1) < 0$, we note that \eqref{eq:Qineq} cannot be an
equality for all $x$, because $d Q_{A^2} (x) - \op{tr} (A) Q_A (x)$ does not vanish
identically unless $A$ is a multiple of the identity matrix (which
corresponds to the cubic case, for which we know the explicit values from
Theorem \ref{thm:jump}). In the non-cubic case, the right-hand side thus is a
nonzero polynomial times the logarithm of a nonconstant polynomial, while the
left-hand side is a polynomial.

Since $\log Q_{\lambda A}(x) = \log\lambda + \log Q_A(x)$, a $\lambda>0$ satisfying \eqref{eq:Qineq}
certainly exists if the sign of $d Q_{A^2} (x) - \op{tr} (A) Q_A (x)$ is
constant for all $x$ with $\|x\|_{\infty} = 1$. We have thus proved the
following result.

\begin{proposition}
  \label{prop:Qineqc}Let $Q$ be a positive definite quadratic form on $\RR^d$ such that
  \begin{equation}
    d \,Q_{A^2} (x) \leq \op{tr} (A) Q_A (x) \label{eq:Qineqc}
  \end{equation}
  for all $x$ with $\|x\|_{\infty} = 1$. Then $V_Q' (d / 2 - 1) < 0$. The same
  conclusion holds if `$\leq$' is replaced with `$\geq$' in
  (\ref{eq:Qineqc}).
\end{proposition}

\begin{example}[Some non-cubic lattices]
  Consider the case when $A$ is given by $A_p \assign I - p E$, where $E$ is the matrix
  with all entries equal to $1$. One easily checks that $A_p$ is positive
  definite if and only if $p < 1 / d$. Hence, we assume $p < 1 / d$. We
  further observe that
  \[ Q_{A_p} (x) = \|x\|_2^2 - p \left( \sum_{j = 1}^d x_j \right)^2, \]
  as well as $A_p^2 = A_{p (2 - d p)}$. Thus equipped, a brief calculation
  reveals that
  \[ d Q_{A_p^2} (x) - \op{tr} (A_p) Q_{A_p} (x) = p d \|x\|_2^2 - p \left[
     1 - (d - 1) p \right] \left( \sum_{j = 1}^d x_j \right)^2 . \]
  Notice that, by H\"older's inequality,
  \[ \left( \sum_{j = 1}^d x_j \right)^2 \leq \|x\|_1^2 \leq d
     \|x\|_{\infty} \|x\|_2^2 . \]
  Assume further that $p \geq 0$, so that $p \left[ 1 - (d - 1) p \right]
  > 0$. We then find that, for all $x$ with $\|x\|_{\infty} = 1$,
  \[ d Q_{A_p^2} (x) - \op{tr} (A_p) Q_{A_p} (x) \geq p^2 d (d - 1)
     \|x\|_2^2 \geq 0. \]
  By Proposition \ref{prop:Qineqc}, we have thus shown that $V_Q' (d / 2 - 1)
  < 0$, with $Q = Q_{A_p}$, for all $0 \leq p < 1 / d$.
  \qede
\end{example}

Continuing in this vein, we explicitly determine $V_Q'(d/2-1)$ for some
very simple binary forms.

\begin{example}
  To indicate the nature of the quantities $V_Q (s)$ and, in consequence,
  $V_Q' (d / 2 - 1)$, let us consider the very basic case of $Q (x_1, x_2) := a
  x_1^2 + b x_2^2$, with $a, b > 0$, (of course, the factor $d / 2 - 1$ in
  (\ref{eq:jumpx}) vanishes in this case, so the contribution of $V'_Q (d
  / 2 - 1)$ is not, in the end, brought to bear). We have
  \begin{eqnarray*}
    V_Q (s) & = & \int_{\|x\|_{\infty} = 1} \frac{(a b - (2 s + 1) a^2) x_1^2
    + (a b - (2 s + 1) b^2) x_2^2}{(a x_1^2 + b x_2^2)^{s + 2}} \mathd
    \lambda_1\\
    & = & 4 \int_0^1 \frac{(a b - (2 s + 1) a^2) x_1^2 + (a b - (2 s + 1)
    b^2)}{(a x_1^2 + b)^{s + 2}} \mathd x_1\\
    &  & + 4 \int_0^1 \frac{(a b - (2 s + 1) a^2) + (a b - (2 s + 1) b^2)
    x_2^2}{(a + b x_2^2)^{s + 2}} \mathd x_2.
  \end{eqnarray*}
  Using the basic integral
  \begin{equation}
    \int_0^1 \frac{1}{(a x^2 + b)^s} \mathd x = \pFq21{1 / 2, s}{3 / 2}{- a},
    \label{eq:intab2F1}
  \end{equation}
  and some standard hypergeometric manipulations, we thus find
  \begin{equation}
    V_Q (s) = \frac{- 8 s}{(a + b)^s} \left[ \pFq21{1, 1 / 2 - s}{3 / 2}{-
      \frac{a}{b}} + \pFq21{1, 1 / 2 - s}{3 / 2}{ - \frac{b}{a}} \right] .
  \end{equation}
  The factor of $s$ in $V_Q (s)$, together with the elementary special case $s
  = 1$ of (\ref{eq:intab2F1}), allows us to conclude that
  \begin{equation*}
    V_Q' (0) = - 8 \left[ \sqrt{\frac{b}{a}} \arctan \sqrt{\frac{a}{b}} +
      \sqrt{\frac{a}{b}} \arctan \sqrt{\frac{b}{a}} \right] .
  \end{equation*}
  In particular, we observe that $V_Q' (d / 2 - 1) < 0$, though Proposition
  \ref{prop:Qineqc} does not apply in the present case.
  \qede
\end{example}

Sadly, as illustrated by this example, Proposition \ref{prop:Qineqc} is not
always accessible. Indeed, it can fail quite comprehensively.

\begin{example}[Some scaled cubic lattices]
  Consider the case when $A$ is given by $A_p \assign I + p D(a)$, where $D(a)=D(a_1,\ldots,a_d)$ is a diagonal matrix
  and, without loss, $p \geq 0$. The matrix $A_p$ is positive
  definite if and only if $p a_k+1 >0$ for all $1 \leq k \leq d$.

  Suppose that $\op{tr} (D(a))=0$, so that $\op{tr} (A_p)=d$. Also $A_p^2 =
  I+2pD(a) +p^2D(a_1^2,\ldots,a_k^2)$. Thence,
  \[ d Q_{A_p^2} (x) - \op{tr} (A_p) Q_{A_p} (x) = pd \sum_{k=1}^d a_k(1+pa_k)x_k^2, \]
  which must change signs on the sphere, since the $a_k$ vary in sign, and so Proposition
  \ref{prop:Qineqc} does not apply.
 \qede
\end{example}

We conclude this section with a comment on the behaviour of $\sigma(s)$ at the
other side of the strip of convergence, that is, as $s \to d/2$.

\begin{remark}[$\sigma(s)$ as $s\to d/2$]\label{rem:spole}
  From Proposition \ref{prop:betares} and the fact that $\alpha_N(s)$ is an
  entire function, we know that $\sigma_N(s)$ has a simple pole at $s=d/2$ with
  the same residue as $\alpha(s)$ (which, by Proposition \ref{prop:sigmastrip},
  is the analytic continuation of the limit $\sigma(s)$).
  \qede
\end{remark}

\section{Alternative procedures for Wigner limits}\label{sec:alt}

The limit $\sigma (s) = \lim_{N \rightarrow \infty} \left[ \alpha_N (s) -
\beta_N (s) \right]$, considered in the preceding sections, is built from the
sum $\alpha_N (s)$, which sums over the lattice points in the hypercube $\{x
\in \RR^d : \|x\|_{\infty} \leq N\}$. In this section, we will
show that some of the previous discussion carries over to the case when the
hypercubes get replaced by more general sets. For simplicity, we restrict to
the case of general hyperballs $\{x \in \RR^d : \|x\| \leq N\}$
where $\| \cdot \|$ is any norm in $\RR^d$, and consider
\begin{eqnarray}
  \widehat{\alpha}_N (s) & := & \sum_{0 < \|n\| \leq N} \frac{1}{Q (n)^s},
  \label{eq:alphah}\\
  \widehat{\beta}_N (s) &: = & \int_{\|x\| \leq N} \frac{1}{Q (x)^s} \mathd x,
  \label{eq:betah}
\end{eqnarray}
as well as $\widehat{\sigma}_N := \widehat{\alpha}_N - \widehat{\beta}_N$. Again, if
$\op{Re} s > d / 2$, then $\widehat{\alpha}_N (s)$ converges to the Epstein zeta
function $\alpha (s) = Z_Q (s)$ as $N \rightarrow \infty$.

Of particular interest is the case $\| \cdot \| = \| \cdot \|_2$, in which the
lattice sum extends over the usual Euclidean $d$-balls of radius $N$. This
case was considered in {\cite[Theorem 2]{latticesums-bbs}} when $d = 2$ and it
was shown that the limit $\widehat{\sigma} (s) \assign \lim_{N \rightarrow \infty}
\widehat{\sigma}_N (s)$ exists in the strip $1 / 3 < \op{Re} s < 1$ and
coincides therein with the analytic continuation of $\alpha (s)$. As we will
see below, this strip can be extended on the left-hand side, though not below
$1 / 4$.

In contrast to Remark \ref{rk:sigma0}, we note that $\widehat{\sigma}_N (0)$
usually does not converge. We therefore let $\lambda$ be the infimum of all
values $\ell \geq 0$ such that
\begin{equation}
  \widehat{\sigma}_N (0) =\# \{n \in \mathbbm{Z}^d : \|n\| \leq N\} -
  \op{vol} \{x \in \RR^d : \|x\| \leq N\} - 1 = O (N^{\ell}) .
  \label{eq:sigmah0}
\end{equation}

The determination of $\lambda$, especially for the $p$-norms $\| \cdot \|_p$,
is a famous problem and in several cases still open. In particular, when \ $d
= 2$ and $\| \cdot \| = \| \cdot \|_2$, this is Gauss's circle problem. For a
recent survey, we refer to {\cite{ikn-lattice04}}. A number of results on the
values of $\lambda$ are discussed in the proof of Corollary
\ref{cor:sigmahstrip2} and the remarks thereafter. We also recall the
well-known fact, due to Weierstrass, that the balls in the $p$-norm have volume
\[ \op{vol} \{x \in \RR^d : \|x\|_p \leq N\} = \frac{2^d
   \Gamma^d (1 + 1 / p)}{\Gamma (1 + d / p)} N^d . \]
We prove the following analog of Proposition \ref{prop:sigmastrip}, which
includes {\cite[Theorem 2]{latticesums-bbs}} as the special case $d = 2$ and
$\| \cdot \| = \| \cdot \|_2$.

\begin{proposition}
  \label{prop:sigmahstrip}Let $\| \cdot \|$ be a norm on $\RR^d$, and
  assume that $\lambda$ is the infimum of all values $\ell \geq 0$ such
  that \eqref{eq:sigmah0} holds. Further, let $Q$ be a positive definite
  quadratic form. Then the limit $\widehat{\sigma} (s) \assign \lim_{N \rightarrow
  \infty} \widehat{\sigma}_N (s)$ exists in the strip $\max (d / 2 - 1, \lambda /
  2) < \op{Re} s < d / 2$ and coincides therein with the analytic
  continuation of $\alpha (s)$.
\end{proposition}

\begin{proof}
  As before, we fix $\sigma > 0$ as well as $R > 0$ and set $\Omega = \{s :
  \op{Re} s > \sigma, \hspace{1em} |s| < R\}$. All order terms below are
  uniform with respect to $s$ in the bounded region $\Omega$. In order to
  proceed along the lines of Proposition \ref{prop:sigmastrip}, we introduce
  \[ \tilde{\beta}_N (s) := \sum_{\|n\| \leq N} \int_{\|x\|_{\infty}
     \leq 1 / 2} \frac{1}{Q (n + x)^s} \mathd x, \]
  and observe that, by \eqref{eq:sigmah0} and the fact that all norms on
  $\RR^d$ are equivalent,
  \begin{equation}
    \widehat{\beta}_N (s) - \tilde{\beta}_N (s) = O (N^{- 2 s + \ell})
    \label{eq:betaht}
  \end{equation}
  for all values $\ell > \lambda$. On the other hand, set $\tilde{\sigma}_N
  (s) := \widehat{\alpha}_N (s) - \tilde{\beta}_N (s)$ and let
  \begin{eqnarray*}
    \delta_N (s) & := & \tilde{\sigma}_N (s) - \tilde{\sigma}_{N - 1} (s)\\
    & = & \sum_{N - 1 < \|n\| \leq N} \int_{\|x\|_{\infty} \leq 1 /
    2} \left[ \frac{1}{Q (n)^s} - \frac{1}{Q (n + x)^s} \right] \mathd x.\\
    & = & O (N^{- 2 \sigma - 2}) \sum_{N - 1 < \|n\| \leq N} 1\\
    & = & O (N^{d - 2 \sigma - 3}),
  \end{eqnarray*}
  where the estimates follow as in the proof of Proposition
  \ref{prop:sigmastrip}.

  Again, we conclude that the series $\delta (s) :=
  \sum_{N = 1}^{\infty} \delta_N (s)$ converges in the half-plane $\op{Re} s
  > d / 2 - 1$ and defines an analytic function therein. By construction,
  \begin{equation}
    \delta (s) = \lim_{N \rightarrow \infty} \left[ \tilde{\sigma}_N (s) +
    \tilde{\beta}_0 (s) \right] . \label{eq:delta1t}
  \end{equation}
  Since $\tilde{\beta}_0 (s)$ is analytic for $\op{Re} s < d / 2$, it
  follows that the limit $\tilde{\sigma} (s) \assign \lim_{N \rightarrow
  \infty} \tilde{\sigma}_N (s)$ exists in the strip $d / 2 - 1 < \op{Re} s <
  d / 2$. In combination with \eqref{eq:betaht}, this shows that the limit
  $\widehat{\sigma} (s)$ exists in the strip $\max (d / 2 - 1, \lambda / 2) <
  \op{Re} s < d / 2$ and equals $\tilde{\sigma} (s)$ therein.

  For the second part of the claim, we proceed as in the proof of Proposition
  \ref{prop:sigmastrip} and observe that, for $\op{Re} s < d / 2$,
  \begin{equation}
    \widehat{\beta}_N (s) = N^{d - 2 s} \int_{\|x\| \leq 1} \frac{1}{Q (x)^s}
    \mathd x = N^{d - 2 s} \widehat{\beta}_1 (s) . \label{eq:betahN1}
  \end{equation}
  We note that Proposition \ref{prop:betares}, with the same proof, also
  applies to $\widehat{\beta}_N$ in place of $\beta_N$. In particular,
  $\widehat{\beta}_N$ and $\widehat{\beta}_0$ have meromorphic continuations to the
  entire complex plane, and the relation induced by \eqref{eq:betahN1} continues to hold.
  For $\op{Re} s > d / 2$,
  \[ \lim_{N \rightarrow \infty} \widehat{\beta}_N (s) = \lim_{N \rightarrow
     \infty} N^{d - 2 s} \widehat{\beta}_1 (s) = 0. \]
  We therefore have, for $\op{Re} s > d / 2$,
  \begin{equation}
    \delta (s) = \lim_{N \rightarrow \infty} \left[ \widehat{\alpha}_N (s) -
    \widehat{\beta}_N (s) + \tilde{\beta}_0 (s) \right] = \alpha (s) +
    \tilde{\beta}_0 (s) . \label{eq:deltaabh}
  \end{equation}
  On the other hand, it follows from \eqref{eq:betaht} and \eqref{eq:delta1t}
  that, for $\op{Re} s < d / 2$,
  \begin{equation}
    \delta (s) = \widehat{\sigma} (s) + \tilde{\beta}_0 (s) . \label{eq:deltasbh}
  \end{equation}
  Since both $\delta (s)$ and $\tilde{\beta}_0 (s)$ are meromorphic in the
  half-plane $\op{Re} s > d / 2 - 1$, comparing \eqref{eq:deltaabh} and
  \eqref{eq:deltasbh} proves that the analytic continuations of $\widehat{\sigma}
  (s)$ and $\alpha (s)$ agree.
\end{proof}

\begin{corollary}[Four and higher dimensions]
  \label{cor:sigmahstrip2}Let $Q$ be a positive definite quadratic form on
  $\RR^d$ for $d \geq 4$. Then the limit
  \[ \widehat{\sigma} (s) = \lim_{N \rightarrow \infty} \left[ \sum_{0 <\|n\|_2
     \leq N} \frac{1}{Q (n)^s} - \int_{\|x\|_2 \leq N} \frac{1}{Q
     (x)^s} \mathd x \right] \]
  exists in the strip $d / 2 - 1 < \op{Re} s < d / 2$ and coincides therein
  with the analytic continuation of $\alpha (s)$.
\end{corollary}

\begin{proof}
  We recall, see {\cite{ikn-lattice04}}, the fact that, for all $d \geq
  5$,
  \[ \# \{n \in \mathbbm{Z}^d : \|n\|_2 \leq N\} - \op{vol} \{x \in
     \RR^d : \|x\|_2 \leq N\} = O (N^{d - 2}), \]
  while, for $d = 4$, the right-hand side needs to be replaced with, for
  instance, the rather classical $O (N^2 (\log N))$, or the improved $O (N^2
  (\log N)^{2 / 3})$ shown in {\cite{walfisz59}}. In any case, we conclude
  that, for all $d \geq 4$, the infimum $\lambda_d$ of all values $\ell
  \geq 0$ such that \eqref{eq:sigmah0} holds, is $\lambda_d = d - 2$. The
  claim therefore follows from Proposition \ref{prop:sigmahstrip}.
\end{proof}

\begin{remark}[Two and three dimensions]
  Thorough reports on the current status of the cases $d = 2$ and $d = 3$,
  missing in Corollary \ref{cor:sigmahstrip2}, can be found in
  {\cite{ikn-lattice04}}. In the case $d = 2$, it was shown by Hardy as well
  as Landau that $\lambda_2 \geq 1 / 2$. While it is believed that in
  fact $\lambda_2 = 1 / 2$, the best currently known bound is $\lambda
  \leq 131 / 208 \approx 0.6298$, obtained in {\cite{huxley-gauss2}}. For
  $d = 3$, it is known that $\lambda_3 \geq 1$ and it is believed that
  $\lambda_3 = 1$, in which case the conclusion of Corollary
  \ref{cor:sigmahstrip2} would also hold for $d = 3$. The smallest currently
  fully proven upper bound is $\lambda_3 \leq 21 / 16 = 1.3125$ from
  {\cite{hb-gauss3}}. \qede
\end{remark}

\begin{remark}[more general $p$-norms]
  Let us briefly note some results and their consequences for more general
  $p$-norms, again referring to {\cite{ikn-lattice04}} for further details and
  missing cases. Let $d \geq 2$. For integers $p > d + 1$ it is known
  that
  \[ \# \{n \in \mathbbm{Z}^d : \|n\|_p \leq N\} - \op{vol} \{x \in
     \RR^d : \|x\|_p \leq N\} = O \left( N^{(d - 1) (1 - 1 / p)}
     \right), \]
  and that the exponent in this estimate cannot be improved. This result was obtained in
  {\cite{randol-gaussp}} for even $p$, and in {\cite{kraetzel-gaussp}} for odd
  $p$. In light of Proposition \ref{prop:sigmahstrip}, we conclude that the
  limit
  \[ \widehat{\sigma} (s) = \lim_{N \rightarrow \infty} \left[ \sum_{0 <\|n\|_p
     \leq N} \frac{1}{Q (n)^s} - \int_{\|x\|_p \leq N} \frac{1}{Q
     (x)^s} \mathd x \right] \]
  exists in the strip $(d - 1) (1 - 1 / p) / 2 < \op{Re} s < d / 2$ and
  coincides therein with the analytic continuation of $\alpha (s)$. We note
  that this strip shrinks to $d / 2 - 1 / 2 < \op{Re} s < d / 2$ as $p
  \rightarrow \infty$. In particular, for $d = 2$, the physically interesting
  value $\widehat{\sigma} (1 / 2)$ always exists and equals $\alpha \left( 1 / 2
  \right)$. \qede
\end{remark}

\appendix
\renewcommand*{\thesection}{\Alph{section}}
\section{Brief review of cubic lattice sums}\label{sec:cubiclatticesums}

The $d$-dimensional cubic lattice sum
\begin{equation}
  Z_d (s) \assign \sum_{n_1, \ldots, n_d}' \frac{1}{(n_1^2 + n_2^2 +
  \cdots + n_d^2)^s} \label{eq:def:cubicsum},
\end{equation}
which converges for $s > d / 2$, is a special case of an Epstein zeta function
as introduced in \eqref{eq:def:epsteinzeta}. As such, the sum $Z_d (s)$
has a meromorphic continuation to the entire complex plane and satisfies the
functional equation
\begin{equation}
  \frac{Z_d (s) \Gamma (s)}{\pi^s} = \frac{Z_d (d/2 - s) \Gamma (d
  / 2 - s)}{\pi^{d / 2 - s}} . \label{eq:cubicfunc}
\end{equation}
The sum $Z_d (s)$ has a simple pole at $s = d / 2$ with residue $\pi^{d
/ 2} / \Gamma (d / 2)$.
We record that the values of $\pi^{d / 2} / \Gamma (d / 2)$, for $d = 1,2,\ldots,6$, are
\begin{equation*}
  1, \quad \pi, \quad 2\pi, \quad \pi^2, \quad \frac43 \pi^2, \quad \frac12 \pi^3.
\end{equation*}
The plots in Figures \ref{fig:alpha4} and \ref{fig:alphas} illustrate these
functions and their properties in small dimensions. Observe the symmetries
around the poles in Figure \ref{fig:alpha4c} and Figure \ref{fig:alphasb}.

\begin{figure}[htb]
  \begin{center}
    \subfigure[{$Z_4(s)$ on $[-2,8]$}\label{fig:alpha4a}]{\includegraphics[width=0.3\textwidth]{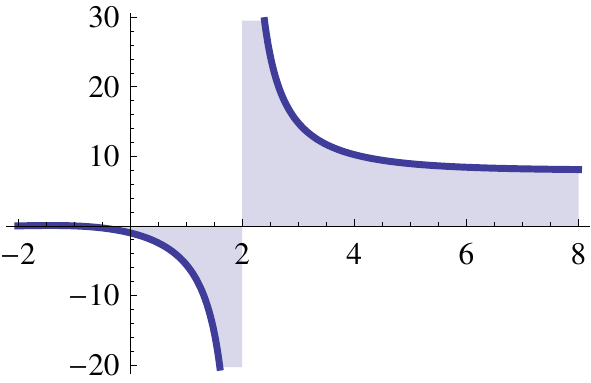}}
    \hfil
    \subfigure[{$Z_4(s)$ on $[-7,0]$}\label{fig:alpha4b}]{\includegraphics[width=0.3\textwidth]{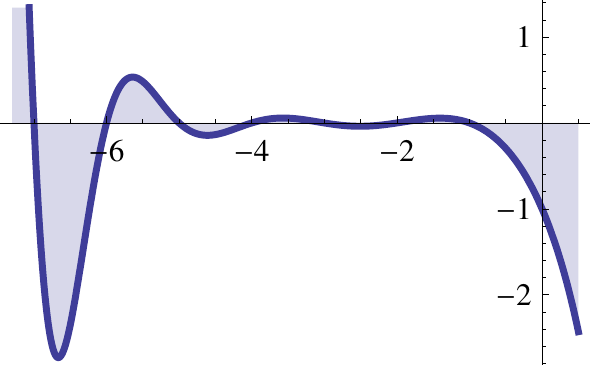}}
    \hfil
    \subfigure[{$Z_4(s) \Gamma(s) \pi^{-s}$}\label{fig:alpha4c}]{\includegraphics[width=0.3\textwidth]{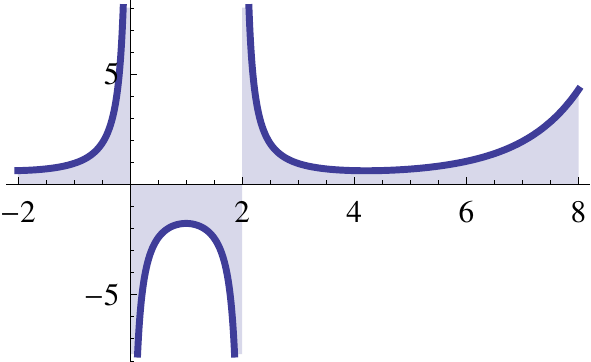}}
    \caption{Different views on $Z_4(s)$ on the real line.}
    \label{fig:alpha4}
  \end{center}
\end{figure}

\begin{figure}[htb]
  \begin{center}
    \subfigure[{$Z_d(s)$ for $d=2,3,4,5$}\label{fig:alphasa}]{\includegraphics[width=0.45\textwidth]{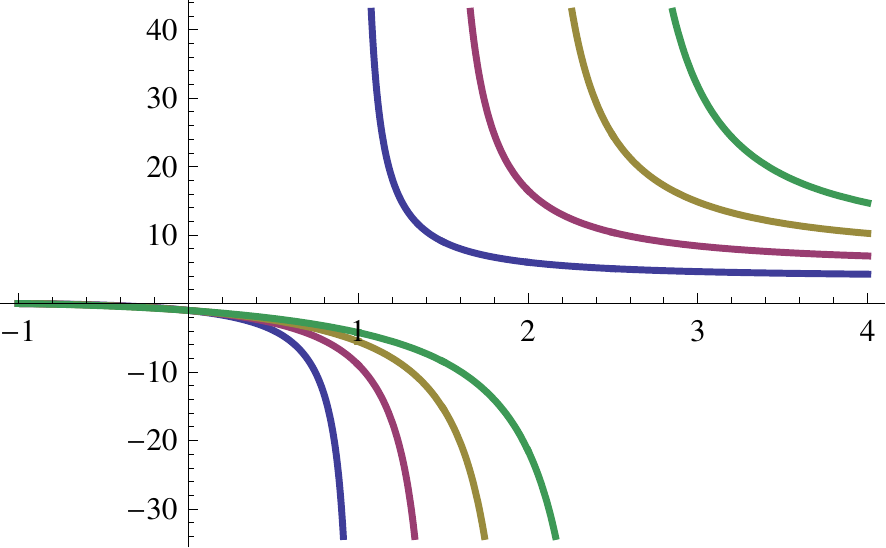}}
    \hfil
    \subfigure[{$Z_d(s) \Gamma(s) \pi^{-s}$ for $d=2,3,4,5$}\label{fig:alphasb}]{\includegraphics[width=0.45\textwidth]{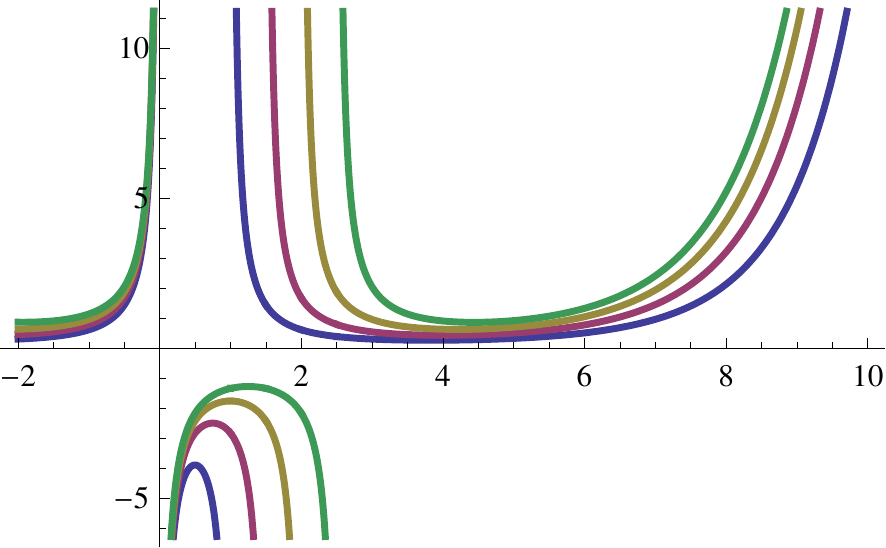}}
    \caption{The functions $Z_d(s)$ on the real line for various values of $d$.}
    \label{fig:alphas}
  \end{center}
\end{figure}

\begin{remark}[Lattice sums and integer representations]
  Let $r_d (n)$ denote the number of integer solutions (counting
  permutations and signs) of $n_1^2 + n_2^2 + \cdots + n_d^2 = n$. Clearly, by
  definition \eqref{eq:def:cubicsum}, the lattice sum $Z_d (s)$ is precisely the Dirichlet series for the sum-of-$d$-squares
  counting function $r_d (n)$, that is,
  \[ Z_d (s) = \sum_{n = 1}^{\infty} \frac{r_d (n)}{n^s} . \]
  In particular, Lagrange's theorem on the sum of four squares
  \cite{latticesums} shows that, if $d>3$, then $r_d(n)>0$ for all $n >0$.
  \qede
\end{remark}

The defining lattice sum \eqref{eq:def:epsteinzeta} only converges when
$\op{Re} s > d / 2$. Our next goal is to make the analytic continuation of
$Z_d (s)$ explicit, in particular in the critical strip $0 \leq
\op{Re} s \leq d / 2$. To this end, let us, for $\op{Re} s > 0$,
define the normalized Mellin transform $M_s [f]$ of a function $f$ on the
positive real line by
\[ M_s [f] := \frac{1}{\Gamma (s)} \int_0^{\infty} f (x) x^{s - 1} \mathd x. \]
The point of this normalization is that, for $\lambda > 0$,
\[ M_s [e^{- \lambda x}] = \frac{1}{\lambda^s} . \]
This allows many classes of lattice sum to be interpreted as the Mellin
transform of combinations of Jacobi theta functions. In the present case of
cubic lattice sums, one finds
\begin{eqnarray*}
  Z_d (s) & = & M_s \left[ \sum_{n_1, \ldots, n_d}' e^{- (n_1^2 + n_2^2
  + \cdots + n_d^2) x} \right]\\
  & = & M_s \left[ \left( \sum_{n = - \infty}^{\infty} e^{- n^2 x} \right)^d
  - 1 \right]\\
  & = & \pi^s M_s \left[ \theta_3^d (i x) - 1 \right],
\end{eqnarray*}
where
\[ \theta_3 (z) := \sum_{n = - \infty}^{\infty} e^{\pi i n^2 z} \]
is the third Jacobi special theta function. In order to obtain the analytic
continuation of $Z_d (s)$, we proceed in the classical fashion and use
the modular transformation
\[ \theta_3 (i / x) = x^{1 / 2} \theta_3 (i x) \]
to write, assuming $\op{Re} s > d / 2$,
\[ \int_0^1 \left( \theta_3^d (i x) - 1 \right) x^{s - 1} \mathd x =
   \frac{1}{s - d / 2} - \frac{1}{s} + \int_1^{\infty} \left( \theta_3^d (i x)
   - 1 \right) x^{d / 2 - s - 1} \mathd x. \]
It follows that
\begin{equation}
  Z_d (s) = \frac{\pi^s}{\Gamma (s)} \left[ \frac{1}{s - d / 2} -
  \frac{1}{s} + \int_1^{\infty} \left( \theta_3^d (i x) - 1 \right) \left(
  x^{s - 1} + x^{d / 2 - s - 1} \right) \mathd x \right] . \label{eq:cubicana}
\end{equation}
We note that the integral in \eqref{eq:cubicana} converges and is analytic for
all $s$. Since the zero of the gamma function cancels the $1 / s$ term, it is
clear from \eqref{eq:cubicana} that $Z_d (s)$ is indeed analytic except
for a simple pole at $s = d / 2$ with residue $\pi^{d / 2} / \Gamma (d / 2)$.
Moreover, the functional equation \eqref{eq:cubicfunc} is another  nearly immediate
consequence of \eqref{eq:cubicana}. Equation \eqref{eq:cubicana} is
well-suited to numerically compute $Z_d (s)$ as well as its analytic
continuation.

\begin{example}
  {\dueto{Exact evaluations}}\label{eg:cubicexact}In small even dimensions,
  the cubic lattice sums can be evaluated in terms of $\zeta (s)$ and $\beta
  (s)$, the Dirichlet series for the primitive character $\chi_{- 4}$ modulo
  $4$. By realizing the lattice sum $Z_d$ as, essentially, the Mellin
  transform of the power $\theta_3^d - 1$, where $\theta_3$ is as before the Jacobi
  theta function, one finds, for instance, the evaluations
  \begin{eqnarray*}
    Z_2 (s) & = & 4 \zeta (s) \beta (s),\\
    Z_4 (s) & = & 8 (1 - 2^{2 - 2 s}) \zeta (s - 1) \zeta (s),\\
    Z_6 (s) & = & 16 \zeta (s - 2) \beta (s) - 4 \zeta (s) \beta (s -
    2),\\
    Z_8 (s) & = & 16 \left( 1 - 2^{1 - s} + 4^{2 - s}
    \right) \zeta (s) \zeta (s - 3 ) .
  \end{eqnarray*}
  See {\cite{latticesums-zucker74}} for these and many further exact
  evaluations of lattice sums. In higher even dimensions, exact evaluations
  involve further $L$-functions. A more direct, but
  equivalent, approach to these evaluations is presented in
  {\cite{bc-dirichlet02}}, where the discussion is based on explicit formulas
  for $r_{2 d} (n)$. For instance, {\cite[Sec.
  6.2]{bc-dirichlet02}},
  \begin{eqnarray*}
    Z_{24} (s) & = & \frac{16}{691} \left( 2^{12 - 2 s} -
    2^{1 - s} + 1 \right) \zeta (s) \zeta (s - 11 )\\
    &  & + \frac{128}{691} \left( 259 + 745 \cdot 2^{4 - s} + 259 \cdot 2^{12
    - 2 s} \right) L_{\Delta} (s),
  \end{eqnarray*}
  where $L_{\Delta}(s) = \sum \tau(n) / n^s$ and $\tau(n)$ is Ramanujan's
  $\tau$-function (here, $\Delta=\eta^{24}$ in terms of the Dedekind $\eta$-function). We remark that the critical
  values of $L_\Delta$ are known to be \emph{periods}, that is, values of an integral
  of an algebraic function over an algebraic domain \cite{periods}. Moreover, up
  to the usual powers of $\pi$, all odd (respectively, even) critical values
  are rational multiples of each other. (More generally, all values
  $L_\Delta(m)$ for integers $m>0$ are \emph{periods}.) Ramanujan's $\tau$ satisfies
  many wonderful congruences. Moreover, Lehmer conjectured \cite{lehmer-tau}
  that $\tau(n)$ (while taking both signs) is never zero, as has been verified
  for more than the first $2 \cdot 10^{19}$ terms \cite{bosman-tau}. Lehmer's
  conjecture is also known to be implied by the alleged irrationality of the
  coefficients of the holomorphic part of a certain Maass-Poincar\'e series,
  see \cite[Example 12.6]{ono-unearthing}.

  The case of odd $d$ is much harder and no simple exact evaluations are
  known. We refer to, for instance, {\cite[Sec. 6]{bc-dirichlet02}} and \cite{latticesums}.
  It transpires in \cite{bc-dirichlet02} that, when considering $r_d(2n)$, the
  smallest odd cases $d=3,5$ are in many ways the hardest in terms of
  estimating asymptotic behavior.
  \qede
\end{example}

\begin{remark}
  [A curious but useful Bessel series] \label{ex:Bessel}

 We recall from {\cite[Sec. 6]{bc-dirichlet02}} a modified Bessel function series for $Z_d$.

 \begin{enumerate}[(a)]
   \item For all integers $d \geq 2$,
     \begin{align}
       Z_d (s) & = 2 d \,\frac{\Gamma \left( (2 s - d + 3) /
       2 \right)}{\Gamma (s+1)} \pi^{(d - 1) / 2}\, \zeta (2 s
       - d + 1 )  \label{eq:cubicbessel}\\
       &\quad + \frac{4 d \pi^{s + 1}}{\Gamma (s+1)} \sum_{m
       \geq 1} \frac{r_{d - 1} (m)}{m^{(d - 2 s - 3) / 4}}
       \sum_{n \geq 1} \frac{K_{(2 s - d + 3) / 2} \left( 2 \pi n \sqrt{m}
       \right)}{n^{(2 s - d - 1) / 2}}. \nonumber
     \end{align}
   \item We note that the first summand of \eqref{eq:cubicbessel}, just like
     the sum $Z_d (s)$ itself, is analytic except for a simple pole at $s = d /
     2$ with residue $\pi^{d / 2} / \Gamma (d / 2)$. Consequently, the double
     sum involving the Bessel terms defines an entire function. Indeed, this
     easily follows directly from the asymptotic fact, see \cite[Chapter 10,
     \S10.40]{DLMF}, that, for positive real argument, the Bessel function
     behaves as
     \[ K_s (x) \sim \sqrt{\frac{\pi}{2 x}} e^{- x} \]
     when $x \rightarrow \infty$. One thus finds that the double sum converges
     for all values of $s$ and defines an analytic function. In particular,
     \eqref{eq:cubicbessel} is another explicit representation of the analytic continuation of
     $Z_d (s)$ to the entire complex plane.

  \item When $r \assign 2s-d+3$ is an odd integer we need compute only Bessel
    function values at half-integers which become elementary
    \cite[\S10.47(ii) and \S10.49(ii)]{DLMF}. When $s = d/2-1$, the
    value of the jump discontinuity, then $r=1$ and we need  consider only
    $K_{1 / 2}\left( 2 \pi n \sqrt{m}\right)$ for integers $m,n>0$. We then
    have \cite[\S10.39 (ii)]{DLMF} that
    \[ K_{1/2}(z) = K_{-1/2}(z) = \sqrt{\frac{\pi}{2z}} e^{-z}. \]
    This reduces \eqref{eq:cubicbessel} to a rapidly convergent exponential
    double series. Summing the second series, for each  positive integer $d>1$, we obtain
    \begin{eqnarray}\nonumber
      Z_d (d / 2 - 1) &=& \frac{2 d \pi^{d / 2}}{\Gamma (d/2)}
      \left[ - \frac{1}{12} + \sum_{m \geq 1} \frac{r_{d - 1}
        (m) e^{- 2 \pi \sqrt{m}}}{(1 - e^{- 2 \pi \sqrt{m}})^2}
        \right]\\&=&\frac{2 d \pi^{d / 2}}{\Gamma (d/2)}
        \left[ - \frac{1}{12} + \frac 12\sum_{m \geq 1} \frac{r_{d - 1}(m)}{\cosh( 2 \pi \sqrt{m})-1} \right]\label{eq:cosh}
    \end{eqnarray}
    for the value corresponding to the jump in Theorem \ref{thm:jump}. Amongst odd integers, \eqref{eq:cosh} is most
    rapid for $d=3$ since $r_2(m)$ is the smallest coefficient set. For even
    integers, \eqref{eq:cosh} combines with Example \ref{eg:cubicexact} to
    provide evaluations of the hyperbolic sum. With $d=2$, this recovers
    \[ \sum_{m=1}^{\infty }\frac 1{ \cosh ( 2\pi m ) -1}
      = \frac1{12}-\frac1{4\pi}, \]
    as the simplest evaluation.
    \qede
  \end{enumerate}
\end{remark}

\section*{Conclusion}

We have been able to analyse the behaviour of Wigner limits for electron sums
in arbitrary dimensions quite extensively. The analysis sheds light on the remarkable interplay
between the physical and analytic properties of lattice sums. We also observe
that physicists typically proceed by taking Laplace and related transforms
quite formally. This suggests that the subtle boundary behaviour of the limit
$\sigma(s)$ would never be noticed without careful mathematical analysis.
Finally, it remains to conclusively answer Problem \ref{prob:VQD} in order to decide
whether, for $d>2$, every quadratic form indeed exhibits a jump at $d/2-1$ in
the corresponding Wigner limit.

\begin{acknowledgements}
The second author was funded by various Australian Research Council grants. The
third author would like to thank the Max-Planck-Institute for Mathematics in
Bonn, where he was in residence when this work was completed, for providing
wonderful working conditions.
\end{acknowledgements}

\newcommand{\etalchar}[1]{$^{#1}$}

\end{document}